\newif\ifreport\reporttrue
\documentclass[conference]{IEEEtran}
\usepackage{setspace}
\usepackage{cite}
\usepackage{graphicx}
\usepackage{bm}
\usepackage[cmex10]{amsmath}
\usepackage{mathtools}
\usepackage{amsthm}
\usepackage{amsfonts}
\usepackage{amssymb}
\usepackage{hyperref}
\usepackage{algpseudocode}
\usepackage[lined, boxed, linesnumbered, ruled]{algorithm2e}
\usepackage{fixltx2e}%fixes latex float algorithm in 2 columns
\usepackage{tikz}
\usetikzlibrary{decorations.pathreplacing}
\theoremstyle{definition}
\newtheorem{definition}{Definition}
\newcommand{\age}{\Delta}
\usepackage{color}

\def\blue{\color{black}}
\def\red{\color{red}}
\newcommand{\ignore}[1]{}

\newtheorem{lemma}{Lemma}

\newtheorem{theorem}{Theorem}

\begin{document}
\IEEEoverridecommandlockouts
%\title{Maximizing Data Freshness for Information Updates}
%\title{Minimizing the Age of Information in Multi-Source Networks}
\title{Minimizing the Age of the Information to Multiple Sources}
\title{\huge	 Age-Optimal  Updates of Multiple Information Flows}

%\title{	 Age-Optimal Status Updates in Multi-flow Multi-server Networks}
\author{Yin Sun$^\dag$, Elif Uysal-Biyikoglu$^\ddag$, and Sastry Kompella$^*$ \\
$^\dag$Dept. of ECE, Auburn University, Auburn, AL\\
$^\ddag$Dept. of EEE, Middle East Technical University, Ankara, Turkey\\
$^*$Information Technology Division, Naval Research Laboratory, Washington, DC
\thanks{Y. Sun's work is supported in part by ONR grant N00014-17-1-2417. S. Kompella's work is supported in part by ONR.}
}

\maketitle
\thispagestyle{plain}
\pagestyle{plain}
% !TEX root = ./sampling_BM.tex
\begin{abstract}
In this paper, we study an age of information minimization problem, where multiple flows of update packets are sent over multiple servers to their destinations. Two online scheduling policies are proposed. When  the packet generation and arrival times  are synchronized across the flows, the proposed policies are shown to be (near) optimal for minimizing any \emph{time-dependent}, \emph{symmetric}, and \emph{non-decreasing} penalty function of the ages of the flows over time in a stochastic ordering sense. 
\end{abstract}

\section{Introduction}
%{\red ideas: where do we have strong need for multi-server multi-flow scheduling? IoT, social networks, online gaming, news, and notifications.} 
%In computer and communcation networks, 

%The  increased availability of network connected mobile devices has spurred a plethora of industrial and daily life applications involving real-time remote measurement, tracking, and control, which relies heavily on the availability of fresh information updates. 
%
%These applications, whether their end user is a person (e.g., a social network or news feed, driving directions) or device (e.g., industrial environmental monitoring, vehicle sensor status, automated driving), are characterized by a dependence on status updates, that is, information packets that contain recently sampled data. Status updates are desired to be sufficiently \emph{fresh}, or \emph{timely} for the application at hand. 

In many information-update and networked control systems, such as news updates, stock trading,  autonomous driving, and robotics control, information has the greatest value when it is fresh. A metric on information freshness, called the \emph{age of information} or simply the \emph{age}, was defined in \cite{Song1990,KaulYatesGruteser-Infocom2012}. Consider a flow of update packets that are sent from a source to a destination through a queue. 
Let $U(t)$ be the time stamp (i.e., generation time) of the {newest update that the destination has received} by time $t$. {The age of information, as a function of time $t$, is defined as} 
$\Delta (t) = t - U(t)$, which is the time elapsed since the newest  update was generated.

In recent years, there have been a lot of research efforts on the behavior of  $\Delta (t)$ under various queueing service disciplines and how to control $\Delta (t)$ to keep the information as fresh as possible
%how to reduce the age $\Delta (t)$ and keep the information fresh, e.g., 
\cite{KaulYatesGruteser-Infocom2012,Suninfocom2016,AgeOfInfo2016,Bedewy2016,BedewyJournal2017,Bedewy2017,BedewyMultihop2017,2012ISIT-YatesKaul,Yates2016, 2015ISITHuangModiano,IgorAllerton2016,HsuTWC2017,He2018,YatesISIT2017}. When there is  a single flow of update packets, a Last Generated
First Served (LGFS) update transmission policy, in which the last generated packet is served the first, has been shown to be (nearly) optimal for minimizing the age process $\{\age(t),t\geq 0\}$ in a stochastic ordering sense for multi-server and multi-hop networks \cite{Bedewy2016,BedewyJournal2017,Bedewy2017,BedewyMultihop2017}. This result holds for arbitrary packet generation times at the source and arbitrary packet arrival times at the transmitter queue (see Fig. \ref{fig_model}); it also holds for minimizing any non-decreasing functional $p(\{\age(t),t\geq 0\})$ of the age process. These studies motivated us to explore service and scheduling policies for achieving age optimality in more general systems with \emph{multiple flows of update packets}. In this case, the transmission scheduler needs to compare not only the packets from the same flow, but also the packets from different flows, which makes the scheduling problem more challenging. 

In this paper, we study age-optimal online scheduling in multi-flow, multi-server queueing systems (as illustrated in Figure \ref{fig_model}), where each server can be used to send update packets to any destination, one packet at a time. {\blue We assume that the packet generation and arrival times are {synchronized} across  the flows. This assumption is a generalized version of  the model in \cite{IgorAllerton2016}.
In practice, synchronized update generations and arrivals occur  when there is a single source and multiple destinations (e.g.,  \cite{IgorAllerton2016}), or in periodic sampling where multiple sources are synchronized by the same clock as in many monitoring and control applications\ifreport
(e.g.,  \cite{Phadke1994,Sivrikaya2004})
\fi.}
The contributions of this paper are summarized as follows:
\begin{itemize}
\item Let $\bm{\Delta}(t)$ denote the age vector of multiple flows. We introduce an age penalty function $p_t (\bm\age(t))$ to represent the level of dissatisfaction for having aged information at the destinations at time $t$, where $p_t$ can be any \emph{time-dependent}, \emph{symmetric}, and \emph{non-decreasing} function of the age vector $\bm{\Delta}(t)$. 

\item For single-server systems with \emph{i.i.d.} exponential service times, we propose a \emph{Maximum Age First, Last Generated First Served (MAF-LGFS) policy}. 
%For synchronized packet generations and arrivals
If the packet generation and arrival times  are synchronized across the flows, then for all age penalty functions $p_t$ defined above, the preemptive MAF-LGFS policy is proven to minimize the age penalty process $\{p_t (\bm\age(t)), t\geq 0\}$ among all causal policies in a stochastic ordering sense (Theorem \ref{thm1}).   

\item For multi-server systems with \emph{i.i.d.} New-Better-than-Used (NBU) service times (which include exponential service times as a special case),
% age-optimal multi-flow online scheduling is quite difficult to achieve. In this case, 
we consider an age lower bound called the \emph{Age of Served Information} and propose a \emph{Maximum Age of Served Information First, Last Generated First Served (MASIF-LGFS) policy}. For synchronized packet generations and arrivals, the non-preemptive MASIF-LGFS policy is shown to be within an  additive gap from the optimum for minimizing the long-run average age of the flows, where the gap is equal to the mean service time of one packet (Theorems \ref{thm3}-\ref{thm4}). Numerical evaluations are provided to verify our (near) age optimality results. {\blue Some possible extensions are discussed at the end of the paper.}

\end{itemize}
\begin{figure}
\centering 
\includegraphics[width=0.4\textwidth]{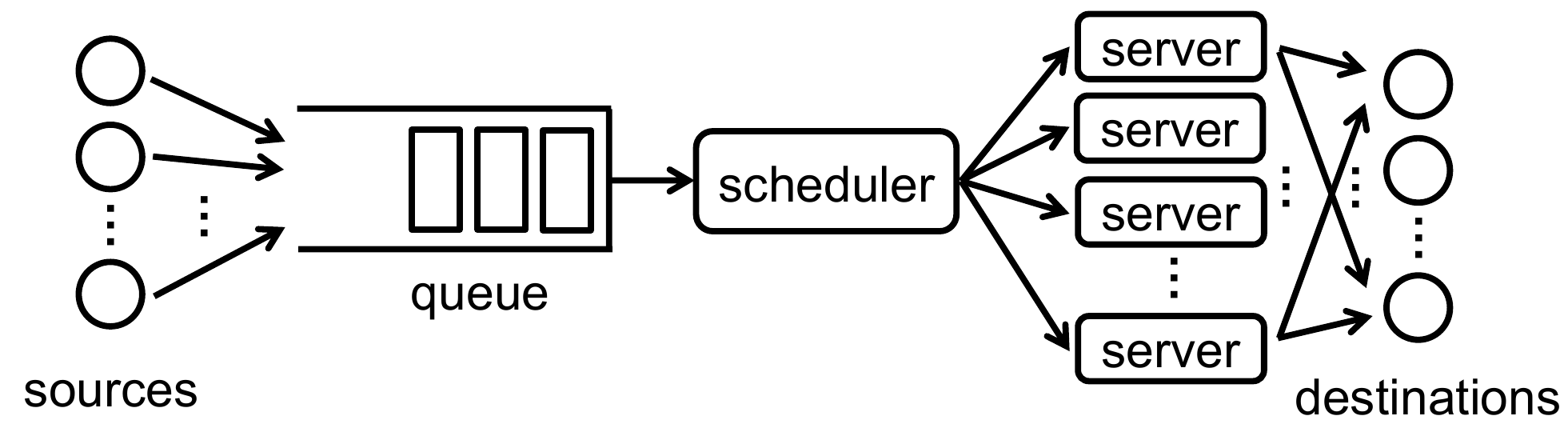} 
\vspace{-0mm}
\caption{System model.}
% work--efficiency ordering holds for any priorities of the jobs.
\label{fig_model} 
\vspace{-5mm}
\end{figure} 
%A comparison with related work is presented in Section \ref{sec_related_work}. 
\ifreport
Our results can be potentially applied to: (i) cloud-hosted Web services where the servers in Figure \ref{fig_model} represent a pool of threads (each for a TCP connection) connecting a front-end proxy node to clients \cite{Fox:1997:CSN:269005.266662}, (ii) industrial robotics and factory automation systems where multiple sensor-output flows are sent to a wireless AP and then forwarded to a system monitor and/or controller \cite{Gungor2009}, and (iii) Multi-access Edge Computing (MEC) that can process fresh data (e.g., data for video analytics, location services, and IoT) locally at the very edge of the mobile network \cite{MEC}. 
\fi

% !TEX root = ./Age_of_Info_multi_source.tex

\section{Related Work}\label{sec_related_work}
The age performance of multiple sources has been analyzed in \cite{2012ISIT-YatesKaul,Yates2016,2015ISITHuangModiano}. In \cite{YatesISIT2017}, status updates over a multiaccess channel was studied. In \cite{He2018}, an age minimization problem for single-hop wireless networks with interference constraints was shown to be NP hard, and tractable cases were identified.  In \cite{IgorAllerton2016}, the expected time-average of the weighted sum age of multiple sources was minimized in a broadcast network with an ON-OFF channel and periodic arrivals, where only one source is scheduled at a time and the scheduler does not know the current ON-OFF channel state. 
When the network is symmetric and the weights are equal, a sample-path method was used to show that the maximum age first (MAF)  policy is optimal. Further, a sub-optimal Whittle's index method was used to handle the general asymmetric cases. In \cite{HsuTWC2017}, for symmetric Bernoulli arrivals and an always-ON channel with no buffers, the MAF policy was shown to be optimal for minimizing  the expected time-average of the sum age of multiple sources. In addition, Markov decision process (MDP) methods were used to handle the general scenarios with asymmetric arrivals and a buffer, where the optimal policies are shown to be switch-type. 

{\blue Compared with these prior studies, Theorem \ref{thm1} in this paper may be seen as an extension of the optimal scheduling results in \cite{IgorAllerton2016,HsuTWC2017} to general time-dependent, symmetric, and non-decreasing age penalty functions $p_t$. In Theorems \ref{thm3}-\ref{thm4}, we go one step forward to study multi-flow, multi-server scheduling, which was not considered in \cite{IgorAllerton2016,HsuTWC2017}. This paper  also complements the studies in
\cite{Bedewy2016,BedewyJournal2017,Bedewy2017,BedewyMultihop2017} on (near) age-optimal online scheduling with a single information flow. }

\section{System Model}\label{sec:model}
\subsection{Notation and Definitions}
% %, with $|\mathcal{S}|$ denoting the cardinality of $\mathcal{S}$.
%%For any random variable ${X}$ and any event $\mathcal{A}$, let $[{X}|\mathcal{A}]$ denote a random variable with the conditional distribution of ${X}$ for given $\mathcal{A}$. 

%For any random variable $Z$ and event $\mathcal{A}$, let $[Z|\mathcal{A}]$ denote a random variable with the conditional distribution of $Z$ for given $\mathcal{A}$.
% and $\mathbb{E}[Z|\mathcal{A}]$ denote the conditional expectation of $Z$ for given $\mathcal{A}$. 
%Let $u$ and $1_A$ denote the unit step function and indicator function of event $A$, respectively, i.e.,
%\begin{align}
%u(t) = \left\{\begin{array}{l l} 1,& \text{if}~t\geq0;\\0,& \text{if}~t<0,\end{array}\right.~~~1_A(x) = \left\{\begin{array}{l l} 1,& \text{if}~x\in A;\\0,& \text{if}~x\notin A.\end{array}\right.\nonumber
%\end{align}

%Let $\bm{x} = (x_1, x_2,\ldots,$ $x_m)$ and $\bm{y} = (y_1, y_2,\ldots,y_m)$ be two vectors in $\mathbb{R}^m$, then we denote $\bm{x} \leq \bm{y}$ if $x_i \leq y_i$ for $i = 1,2,\ldots,m$. A set $U \subseteq \mathbb{R}^m$ is called \emph{upper}, if for all $\bm{x} \in U$ and $\bm{y}\geq \bm{x}$ it holds that $\bm{y} \in U$. 

We  use lower case letters such as $x$ and $\bm{x}$, respectively, to represent deterministic scalars and vectors. In the vector case, a subscript will index the components of a vector, such as $x_i$.
%We use $x_{[i]}$ %and $x_{(i)}$, respectively, 
%to denote the $i$-th largest %and the $i$-th smallest 
%component of $\bm{x}$. 
We use $x_{[i]}$ %and $x_{(i)}$, respectively, 
to denote the $i$-th largest %and the $i$-th smallest 
component of vector $\bm{x}$. Let $\bm 0$ denote   the vector with all 0 components.
%Let $\bm{x}_{\uparrow}=(x_{(1)},\ldots,x_{(n)})$ %and $\bm{x}_{\downarrow}=(x_{[1]},\ldots,x_{[n]})$, respectively,
%denote the increasing %and decreasing
%rearrangements of $\bm{x}$. 
A function $f: \mathbb{R}^n\rightarrow \mathbb{R}$ is termed \emph{symmetric} if $f(\bm{x})= f(x_{[1]},\ldots, x_{[n]})$ for all $\bm{x}$. A function $f: \mathbb{R}^n\rightarrow \mathbb{R}$ is termed \emph{separable} if there exists functions $f_1,\ldots,f_n$ of one variable such that $f(\bm{x}) = \sum_{i=1}^n f_i(x_i)$. 
The composition of functions $f$ and $g$ is denoted by $f \circ g(\bm x) = f(g (\bm x))$. 
 For any $n$-dimensional vectors $\bm{x}$ and $\bm{y}$, the elementwise vector ordering $x_i\leq y_i$, $i=1,\ldots,n$, is denoted by $\bm{x} \leq \bm{y}$. 
 %Further, $\bm{x}$ is said to be \emph{majorized} by $\bm{y}$, denoted by $\bm{x}\prec\bm{y}$, if (i) $\sum_{i=1}^j x_{[i]} \leq \sum_{i=1}^j y_{[i]}$, $j=1,\ldots,n-1$ and (ii) $\sum_{i=1}^n x_{[i]} = \sum_{i=1}^n y_{[i]}$ \cite{Marshall2011}. In addition, $\bm{x}$ is said to be  \emph{weakly majorized by $\bm{y}$ from below}, denoted by $\bm{x}\prec_{\text{w}}\bm{y}$, if $\sum_{i=1}^j x_{[i]} \leq \sum_{i=1}^j y_{[i]}$, $j=1,\ldots,n$; $\bm{x}$ is said to be  \emph{weakly majorized by $\bm{y}$ from above}, denoted by $\bm{x}\prec^{\text{w}}\bm{y}$, if $\sum_{i=1}^j x_{(i)} \geq \sum_{i=1}^j y_{(i)}$, $j=1,\ldots,n$ \cite{Marshall2011}.
%A function that preserves the majorization order is called a Schur convex function. Specifically, $f: \mathbb{R}^n\rightarrow \mathbb{R}$ is termed \emph{Schur convex} if $f(\bm{x})\leq f(\bm{y})$ for all $\bm{x}\prec\bm{y}$ \cite{Marshall2011}.
Let $\mathcal{A}$ 
and $\mathcal{U}$ 
denote sets and events. 
For all random variable ${X}$ and event $\mathcal{A}$, let $[{X}|\mathcal{A}]$ denote a random variable with the conditional distribution of ${X}$ for given $\mathcal{A}$.

%\begin{definition}\emph{Majorization \cite{Marshall2011}}: For any $n$-dimensional vectors $\bm{x}$ and $\bm{y}$, $\bm{x}$ is said to be \emph{majorized} by $\bm{y}$, denoted by $\bm{x}\prec\bm{y}$, if (i) $\sum_{i=1}^j x_{[i]} \leq \sum_{i=1}^j y_{[i]}$, $j=1,\ldots,n-1$ and (ii) $\sum_{i=1}^n x_{[i]} = \sum_{i=1}^n y_{[i]}$. In addition, $\bm{x}$ is said to be  \emph{weakly majorized by $\bm{y}$ from below}, denoted by $\bm{x}\prec_{\text{w}}\bm{y}$, if $\sum_{i=1}^j x_{[i]} \leq \sum_{i=1}^j y_{[i]}$, $j=1,\ldots,n$.
%% $\bm{x}$ is said to be  \emph{weakly majorized by $\bm{y}$ from above}, denoted by $\bm{x}\prec^{\text{w}}\bm{y}$, if $\sum_{i=1}^j x_{(i)} \geq \sum_{i=1}^j y_{(i)}$, $j=1,\ldots,n$ . 
%\end{definition}
%
%\begin{definition}\emph{Schur Convexity \cite{Marshall2011}}:  A function that preserves the majorization order is called a Schur convex function. Specifically, $f: \mathbb{R}^n\rightarrow \mathbb{R}$ is termed \emph{Schur convex} if $f(\bm{x})\leq f(\bm{y})$ for all $\bm{x}\prec\bm{y}$ \cite{Marshall2011}. 
%\end{definition}
% Define $x\wedge y=\min\{x,y\}$.

\begin{definition}\label{def_variable}
\emph{Stochastic Ordering of Random Variables \cite{StochasticOrderBook}}: 
A random variable ${X}$ is said to be \emph{stochastically smaller} than another random variable ${Y}$, denoted by ${X}\leq_{\text{st}}{Y}$, if 
\begin{align}
\Pr({X}>t) \leq \Pr({Y}>t),~\forall~t\in \mathbb{R}. \nonumber
\end{align}
\end{definition}
\begin{definition}\label{def_vector}
\emph{Stochastic Ordering of Random Vectors \cite{StochasticOrderBook}}: 
A set $\mathcal{U} \subseteq \mathbb{R}^n$ is called \emph{upper}, if $\bm{y} \in \mathcal{U}$ whenever $\bm{y}\geq \bm{x}$ and $\bm{x} \in \mathcal{U}$. 
Let $\bm{X}$ and $\bm{Y}$ be two $n$-dimensional random vectors, $\bm{X}$ is said to be \emph{stochastically smaller} than $\bm{Y}$, denoted by $\bm{X}\leq_{\text{st}}\bm{Y}$, if 
\begin{align}
\Pr(\bm{X}\in \mathcal{U}) \leq \Pr(\bm{Y}\in \mathcal{U}),~\forall~\mathcal{U}\subseteq \mathbb{R}^n.\nonumber
\end{align}
\end{definition}
\begin{definition}\label{def_process}
\emph{Stochastic Ordering of Stochastic Processes \cite{StochasticOrderBook}}: 
Let $\{X(t),t\in [0,\infty) \}$ and $\{Y(t),t\in [0,\infty) \}$ be two stochastic processes, $\{X(t), t\in [0,\infty) \}$ is said to be \emph{stochastically smaller} than $\{Y(t),t\in [0,\infty) \}$, denoted by $\{X(t),t\in [0,\infty) \}\leq_{\text{st}}\{Y(t),t\in [0,\infty)\}$, if for all integer $n$ and $0\leq t_1< t_2<\ldots<t_n$, it holds that 
\begin{align}
\!\!\!(X(t_1),X(t_2),\ldots,X(t_n)) \!\leq_{\text{st}}\! (Y(t_1),Y(t_2),\ldots,Y(t_n)).\!\!\!\nonumber
\end{align}
\end{definition}

\ignore{A random variable ${X}$ is said to be \emph{stochastically smaller} than another random variable ${Y}$, denoted by ${X}\leq_{\text{st}}{Y}$, if $\Pr({X}>x) \leq \Pr({Y}>x)$ for all~$x\in \mathbb{R}$.
A set $\mathcal{U} \subseteq \mathbb{R}^n$ is called \emph{upper}, if $\bm{y} \in \mathcal{U}$ whenever $\bm{y}\geq \bm{x}$ and $\bm{x} \in \mathcal{U}$. 
A random vector $\bm{X}$ is said to be \emph{stochastically smaller} than another random vector $\bm{Y}$, denoted by $\bm{X}\leq_{\text{st}}\bm{Y}$, if $\Pr(\bm{X}\in \mathcal{U}) \leq \Pr(\bm{Y}\in \mathcal{U})$ for all upper sets ~$\mathcal{U}\subseteq \mathbb{R}^n$. %If $\bm{X}\leq_{\text{st}}\bm{Y}$ and $\bm{X}\geq_{\text{st}}\bm{Y}$, then $\bm{X}$ and $\bm{Y}$ follow the same distribution, denoted by $\bm{X}=_{\text{st}}\bm{Y}$. 
A stochastic process $\{X(t), t\in [0,\infty) \}$ is said to be \emph{stochastically smaller} than another stochastic process $\{Y(t),t\in [0,\infty) \}$, denoted by $\{X(t),t\in [0,\infty) \}\leq_{\text{st}}\{Y(t),t\in [0,\infty)\}$, if for all integer $n$ and $0\leq t_1< t_2<\ldots<t_n$, it holds that $(X(t_1),X(t_2),\ldots,X(t_n)) \!\leq_{\text{st}}\! (Y(t_1),Y(t_2),\ldots,Y(t_n))$.}
%\begin{align}
%\!\!\!.\!\!\!\nonumber
%\end{align}

\ignore{
Let us use $\mathcal{U}$ % and $\mathcal{A}$ 
to denote sets. A set $\mathcal{U} \subseteq \mathbb{R}^n$ is called \emph{upper}, if $\bm{y} \in \mathcal{U}$ whenever $\bm{y}\geq \bm{x}$ and $\bm{x} \in \mathcal{U}$. % or events.

}

Let $\mathbb{V}$ be the set of Lebesgue measurable functions on $[0,\infty)$, i.e.,
\begin{align}\label{eq_functions}
\mathbb{V} = \{f : [0,\infty) \mapsto \mathbb{R} \text{ is Lebesgue measurable}\}.
\end{align}
A functional $\phi:\mathbb{V}\mapsto\mathbb{R}$ is said to be  \emph{non-decreasing} if  $\phi(f_1) \leq \phi(f_2)$ {for all $f_1,f_2\in\mathbb{V}$ satisfying} $f_1(t)\leq  f_2(t)$ for $t\in [0,\infty)$.
We remark that $\{X(t),t\in [0,\infty) \}\leq_{\text{st}}\{Y(t),t\in [0,\infty)\}$  if, and only if,  \cite{StochasticOrderBook}
\begin{equation}\label{eq_order}
\mathbb{E}[\phi(\{X(t), t\in [0,\infty)\} )] \leq \mathbb{E}[\phi(\{Y(t), t\in [0,\infty)\} )]
\end{equation}
holds for all non-decreasing functional $\phi: \mathbb{V}\rightarrow \mathbb{R}$, provided that the expectations in \eqref{eq_order} exist.

\subsection{Queueing System Model}
%{\red Can be generalized to multiple servers.}

Consider the status update system that is illustrated in Fig. \ref{fig_model}, where $N$ flows of update packets are sent through a queue with $M$ servers and an infinite buffer. Let $s_n$ and $d_n$ denote the source and destination nodes of flow $n$, respectively. Different flows can have different source and/or destination nodes. Each packet can be assigned to any server, and a server can only process one packet at a time. The service times of the update packets are  \emph{i.i.d.} across the servers and time. %Hence, the packet service time distribution depends on the server, rather than the packet. 

%In practice, the servers can be TCP/HTTP connections or wireless communication channels that can be assigned to different flows. 
% {\red and a buffer size $B\geq N$}
% infinite buffer size.\footnote{It is easy to check that the results in this paper hold if the buffer size $B$ is finite and $B\geq N$.} 

% a transmitter sends update packets to $R$ receivers through $M$ communication servers, where a server could be a wireless communication channel, a TCP/HTTP connection, etc. 
The system starts to operate at time $t=0$.  The $i$-th update packet of flow $n$ is generated  at the source node $s_n$ at time $S_{n,i}$, arrives at the queue at time $A_{n,i}$, and is delivered to the destination $d_n$ at time $D_{n,i}$ such that $0\leq S_{n,1} \leq S_{n,2}\leq\ldots$ and $S_{n,i}\leq A_{n,i}\leq D_{n,i}$.\ignore{
We consider two types of packet arrival processes: 
\begin{definition} \emph{Synchronized Arrivals:}
The packet arrival times are said to be \emph{synchronized} across the $N$ flows, if there exist a sequence $\{S_1, S_2,\ldots\}$ such that $S_{n,i} = S_i$ for all  $i=1,2,\ldots$ and $n=1,\ldots,N$. 
\end{definition}
\begin{definition} \emph{Periodic Arrivals:}
The packet arrival times are said to be \emph{periodic}, if there exist a period $T>0$ such that $S_{n,i} = iT $ for all  $i=1,2,\ldots$ and $n=1,\ldots,N$. This is a special case of synchronized arrivals.
\end{definition}

{\red out of order arrival, remove periodic arrivals} 
In practice, synchronized arrivals occur when there is a single source and multiple destinations, while periodic arrivals occur when the sources are synchronized by the same clock as in many monitoring and control applications, e.g., \cite{Phadke1994,Mainwaring:2002:WSN:570738.570751,Sivrikaya2004}. {\red Comparison with Ahmed's paper} More general arrival processes will be studied later.
}
We consider the following class of \emph{synchronized} packet generation and arrival processes:
%We assume that the packet generation and arrival times are \emph{synchronized} across the $N$ flows, as defined below:

\begin{definition} \emph{Synchronized Sampling and Arrivals:}
The packet generation and arrival times are said to be \emph{synchronized} across  the $N$ flows, if there exist two sequences $\{S_1, S_2,\ldots\}$ and $\{A_1,A_2,\ldots\}$ such that for all $i=1,2,\ldots,$ and $n=1,\ldots,N$
\begin{align}
S_{n,i} = S_i,~A_{n,i} = A_i.
\end{align}
\end{definition}
 Note that in this paper, the sequences $\{S_1, S_2,\ldots\}$ and $\{A_1,A_2,\ldots\}$ are \emph{arbitrary}. Hence,
\emph{out-of-order arrivals}, e.g., $S_i < S_{i+1}$ but $A_i > A_{i+1}$, are allowed. In addition, when there is a single flow ($N=1$),  synchronized sampling and arrivals reduce to  {arbitrary} packet generation and arrival processes that were considered in \cite{Bedewy2016,BedewyJournal2017,Bedewy2017,BedewyMultihop2017}. %More general arrival processes will be studied later.

Let $\pi$ represent a scheduling policy that determines the packet being sent by the servers over time. Let $\Pi$ denote the set of \emph{causal} policies in which the scheduling decisions are made based on the history and current states of the system.
A policy is said to be \emph{preemptive}, if each server can switch to send another packet at any time; the preempted packet will be stored back to the queue, waiting to be sent at a later time.
A policy is said to be \emph{non-preemptive}, if each server must complete sending the current packet before starting to serve another packet. A policy is said to be \emph{work-conserving}, if all servers are  kept busy whenever  the queue is non-empty.  We use $\Pi_{np}$ to denote the set of non-preemptive causal policies such that $\Pi_{np}\subset \Pi$. 
%and use $\Pi_{npwc}$ to denote the set of non-preemptive, work-conserving, and causal policies, such that $\Pi_{npwc} \subset\Pi_{np}\subset \Pi$.
Let 
\begin{align}
\mathcal{I}=\{S_{i}, A_{i},~ i=1,2,\ldots\} 
\end{align}
denote the packet generation and arrival times of the flows. 
We assume that the packet generation/arrival times $\mathcal{I}$ and the packet service times are
determined by two \emph{mutually independent} external processes, both of which do not change according to the adopted scheduling policy.

\subsection{Age  Metrics}
At any time $t\geq0$, the freshest packet delivered to the destination node $d_n$ is generated at time  
\begin{align}
U_{n} (t) \!=\! \max\{S_{n,i}\!:\! D_{n,i} \leq t, i\!=\!1,2,\ldots\}.
\end{align}
The \emph{age of information}, or simply the \emph{age}, of flow $n$ is defined as \cite{Song1990,KaulYatesGruteser-Infocom2012}
\begin{align}\label{eq_age}
\Delta_{n} (t) = t - U_{n} (t),
\end{align}
which is the time difference between  the current time $t$ and the generation time of the freshest packet currently available at destination $d_n$. Let $\bm{\Delta}(t)=(\Delta_{1} (t),\ldots,\Delta_{N} (t))$ denote the age vector of the $N$ flows at time $t$. 
%represents the staleness of the information available at node $d_n$.
%\begin{align}
%U_{n}(t) = \max\{S_{n,i}: D_{n,i} \leq t\}.
%\end{align}
%to denote the time-stamp of the freshest update packet received by  up to time $t$. 
%The \emph{age of information}, or simply the \emph{age}, of flow $n$ is defined as
%\begin{align}\label{eq_age}
%\Delta_{n} (t) = t - U_{n}(t),
%\end{align}
%which represents the staleness of the available information at  flow $n$. 

We introduce an \emph{age penalty function} $p(\bm{\Delta}) = p\circ \bm{\Delta}$  to represent the level of dissatisfaction for having aged information at the $N$ destinations, where $p: \mathbb{R}^N\rightarrow \mathbb{R}$ can be any  \emph{non-decreasing} function  of the $N$-dimensional age vector  $\bm{\Delta}$. Some  examples of the age penalty function are: %in $\mathcal{P}_{\text{Sch}}$ are:
\begin{itemize}
\item[1.] The \emph{average age} of the $N$ flows is
\begin{align}\label{eq_avgage}
p_{\text{avg}} (\bm{\Delta}) = \frac{1}{N}\sum_{n=1}^N \Delta_{n}. 
\end{align}

\item[2.] The \emph{maximum age} of the $N$ flows is
\begin{align}%\label{eq_maxage}
p_{\max} (\bm{\Delta}) = \max_{n=1,\ldots,N} \Delta_{n}.
\end{align}

\item[3.] The \emph{mean square age} of the $N$ flows is
\begin{align}%\label{eq_msage}
p_{\text{ms}} (\bm{\Delta}) = \frac{1}{N}\sum_{n=1}^N (\Delta_{n} )^2.
\end{align}

\item[4.] The \emph{$l$-norm of the age vector} of the $N$ flows is
\begin{align}%\label{eq_msage}
p_{\text{$l$-norm}} (\bm{\Delta}) = \left[\sum_{n=1}^N (\Delta_{n} )^l\right]^{\frac{1}{l}}, ~l\geq1.
\end{align}

%\item[5.] The \emph{proportional fair age utility} of the flows is
%\begin{align}%\label{eq_msage}
%\Delta_{\text{PF}} (t) = \sum_{n=1}^N \log[\Delta_{n} (t)+ \epsilon],\nonumber
%\end{align}
%where $\epsilon>0$ is any positive number.\footnote{We are }

\item[5.] The \emph{sum age penalty function} of the $N$ flows is
\begin{align}%\label{eq_msage}
p_{\text{sum-penalty}} (\bm{\Delta}) = \sum_{n=1}^N g(\Delta_{n}),
\end{align}
where $g: [0,\infty) \rightarrow \mathbb{R}$ is the age penalty function for each flow, which can be any \emph{non-decreasing} function of the age $\Delta$ of the flow \cite{Suninfocom2016,AgeOfInfo2016}. For example, a stair-shape function $g_1(\Delta)=\lfloor a \Delta\rfloor$ with $a\geq 0$ can be used to characterize the dissatisfaction of data staleness when the information of interests is checked periodically, and an exponential function $g_2(\Delta) = e^{a \Delta}$ is appropriate for online learning and control applications where the desire for information refreshing grows quickly with respect to the age \cite{AgeOfInfo2016}. \end{itemize}

In this paper, we consider a class of \emph{symmetric} and \emph{non-decreasing} age penalty functions, i.e.,
\begin{align}%\label{eq_class}
\mathcal{P}_{\text{sym}}
\!=\!\{p: [0,\infty)^N\rightarrow \mathbb{R}  \text{ is symmetric and non-decreasing}\}.\nonumber
\end{align}
This is a fairly large class of age penalty functions, where the function   $p$ can be discontinuous, non-convex, or non-separable.
It is easy to see 
\begin{align}
\{p_{\text{avg}},p_{\max}, p_{\text{ms}},p_{\text{$l$-norm}}, %\Delta_{\text{PF}}, 
p_{\text{sum-penalty}}\}\subset \mathcal{P}_{\text{sym}}.\nonumber
\end{align}
%Notice that $p(\bm{\Delta})$ is a function of time $t$ and policy $\pi$. 
Note that the age vector $\bm{\Delta}$ is a function of time $t$ and policy $\pi$, and the age penalty function $p$ may change over time. 
We use $\{p_t \circ \bm{\Delta}_\pi(t), t\in[0,\infty)\}$ to represent the stochastic process generated by the \emph{time-dependent} age penalty function $p_t$ in policy $\pi$. We assume that the initial age $\bm\Delta_{\pi}(0^-)$ at time $t=0^-$ remains the same for all $\pi\in\Pi$.

\section{Multi-flow Update Scheduling}\label{sec_analysis}
In this section, we investigate update scheduling of multiple information flows. We first consider a system setting with a single server and exponential service times, where an age optimality result is established. Next, we study a more general system setting with multiple servers and NBU service times. In this case, age optimality is inherently difficult to achieve and we present a near age-optimal result.

\subsection{Multiple Flows, Single Server, Exponential Service Times}
%When there is a single flow, the scheduler needs to  decide which packet to serve the first. In this case, it is known that the  \emph{Last Generated First Served (LGFS)}  scheduling discipline can achieve the minimum age process in a stochastic ordering sense \cite{Bedewy2016,BedewyJournal2017,Bedewy2017,BedewyMultihop2017}. 
%
%When there are multiple flows, the scheduler needs to compare not only the packets from the same flow, but also the packets from different flows, which makes the scheduling problem more complicated. 

To address the multi-flow online scheduling problem, we consider a flow selection discipline called \emph{Maximum Age First  (MAF)}  \cite{LiInfocom2015,IgorAllerton2016,HsuTWC2017}, in which 
\emph{the flow with the maximum age is served the first, with ties broken arbitrarily}. 
%(ii) The second discipline is called \emph{Maximum Age Reduction (MAR) first}: If the $i$-th packet of flow $n$ is fresher than any packet available at the destination $d_n$, the age reduction brought by the delivery of this packet is $S_{n,i}-U_n(t)$. In the MAR discipline,\emph{the packet with the maximum  age reduction is served the first, with ties broken arbitrarily}.
A scheduling policy is defined by combining the MAF and LGFS disciplines as follows:

\begin{definition} \emph{Maximum Age First, Last Generated First Served (MAF-LGFS) policy:} In this policy, the last generated packet from the flow with the maximum age is served the first among all packets of all flows, with ties broken arbitrarily.
%each packet assigned to the server is the last generated packet from the flow with the maximum age, , with ties broken arbitrarily.
%The scheduler first picks the last generated packet from each flow and assign these packets to idle servers according to the MA discipline; if there exist idle servers after the first round, the scheduler picks the second last generated packet from each flow and assign these packets to idle servers according to the MA discipline; this procedure continues until all servers are busy or all packets are under service. 
\end{definition}

%In the special case that there is a single flow ($N=1$), the MAF-LGFS policy reduces to the LGFS policy studied in \cite{Bedewy2016,BedewyJournal2017,Bedewy2017,BedewyMultihop2017}. 

%Hence, in the MAF-LGFS policy, the  LGFS discipline is adopted with a higher priority than the MA discipline. 
%The packet service priorities of the MAF-LGFS policy are illustrated in Fig. \ref{fig_policies}.  
The age optimality of the preemptive MAF-LGFS policy is established in the following theorem.

\begin{theorem}\label{thm1}
If (i) there is a single server ($M=1$), (ii) the packet generation and arrival times are \emph{synchronized} across the $N$ flows, and (iii) the packet service times are exponentially distributed and \emph{i.i.d.} across  time, then it holds that for all $\mathcal{I}$, all $p_t \in\mathcal{P}_{\text{sym}}$, and all $\pi\in\Pi$ 
\begin{align}\label{thm1eq1}
&[\{p_t \circ\bm{\Delta}_{\text{prmp, MAF-LGFS}}(t), t\geq 0\}\vert\mathcal{I}] \nonumber\\
\leq_{\text{st}} & [\{p_t \circ\bm{\Delta}_\pi(t), t\geq 0\}\vert\mathcal{I}],
\end{align}
or equivalently, for all $\mathcal{I}$, all $p_t \in\mathcal{P}_{\text{sym}}$, and all non-decreasing functional $\phi:\mathbb{V}\mapsto\mathbb{R}$
 \begin{align}\label{thm1eq2}
&\mathbb{E}\left[\phi (\{p_t \circ\bm{\Delta}_{\text{prmp, MAF-LGFS}}(t),t\geq 0\})\vert\mathcal{I}\right] \nonumber\\
= & \min_{\pi\in\Pi} \mathbb{E}\left[\phi (\{p_t \circ\bm{\Delta}_\pi(t),t\geq0\})\vert\mathcal{I}\right],
\end{align}
provided that the expectations in \eqref{thm1eq2} exist, where $\mathbb{V}$ is the set of Lebesgue measurable functions defined in \eqref{eq_functions}.
\end{theorem}
\begin{proof}[Proof idea]
%We develop a  sample-path method to prove Theorem \ref{thm1}.
If the packet generation and arrival times are synchronized across the flows, the preemptive MAF-LGFS policy satisfies the following property: When a packet is delivered to its destination, the flow with the maximum age before the packet delivery will have the minimum age among the $N$ flows once the packet is delivered.\footnote{Note that this property does not hold when   packet generations and arrivals are asynchronized.} This is one key idea used in the proof. See Appendix \ref{app1} for the details.  \end{proof}

Theorem \ref{thm1} tells us  that, for all  age penalty functions in $\mathcal{P}_{\text{sym}}$, all number of flows $N$, and all synchronized packet generation and arrival times $\mathcal{I}$, the preemptive MAF-LGFS policy  minimizes the stochastic process $\{p_t \circ\bm{\Delta}_\pi(t), t\geq 0\}$ among all causal policies in a stochastic ordering sense. {\blue We note that a weaker version of Theorem \ref{thm1} is to consider the mixture over the realizations of $\mathcal{I}$, where \eqref{thm1eq1} becomes
\begin{align}
&[\{p_t \circ\bm{\Delta}_{\text{prmp, MAF-LGFS}}(t), t\geq 0\}] \leq_{\text{st}} [\{p_t \circ\bm{\Delta}_\pi(t), t\geq 0\}],\nonumber
\end{align} 
and similarly, the condition $\mathcal{I}$ in \eqref{thm1eq2}  can be removed.}

\subsection{Multiple Flows, Multiple Servers, NBU Service Times}
Next, we consider a more general system setting with multiple servers and a class of New-Better-than-Used (NBU)  service time distributions that include exponential distribution as a special case. 
%The next question we proceed to answer is whether for an important class of distributions that are more general than exponential, age-optimality or near age-optimality can be achieved. 

%NBU distributions are defined as follows.
\begin{definition}  \emph{New-Better-than-Used Distributions:} Consider a non-negative random variable $X$ with complementary cumulative distribution function (CCDF) $\bar{F}(x)=\Pr[X>x]$. Then, $X$ is said to be \emph{New-Better-than-Used (NBU)} if for all $t,\tau\geq0$
\begin{equation}\label{NBU_Inequality}
\bar{F}(\tau +t)\leq \bar{F}(\tau)\bar{F}(t).
\end{equation} 
\end{definition}
Examples of NBU distributions include constant service time, exponential distribution, shifted exponential distribution, geometrical distribution, Erlang distribution, negative binomial distribution, etc. 

In  scheduling literature, optimal online scheduling has been successfully established for delay minimization in single-server queueing systems, e.g., \cite{Schrage68,Jackson55}, but can become inherently difficult in the multi-server cases. {\blue In particular, minimizing the average delay in deterministic scheduling problems with more than one servers is NP-hard  \cite{Leonardi:1997}. Similarly, delay-optimal stochastic scheduling in multi-class, multi-server queueing systems is deemed to be notoriously difficult \cite{Weiss:1992,Weiss:1995,Dacre1999}. The key challenge in multi-class, multi-server scheduling is that \emph{one cannot combine the resources of all the servers to jointly process the most important packet}. Due to the same reason,} age-optimal online scheduling is quite challenging in multi-flow, multi-server systems. In the sequel, we consider a slightly relaxed goal to seek for \emph{near} age-optimal online scheduling of multiple information flows, where our proposed scheduling policy is shown to be within a small additive gap from the optimum age performance.   

%We first construct a lower bound of the age $\age_{n}(t)$: 

\begin{figure}
\centering 
\includegraphics[width=0.25\textwidth]{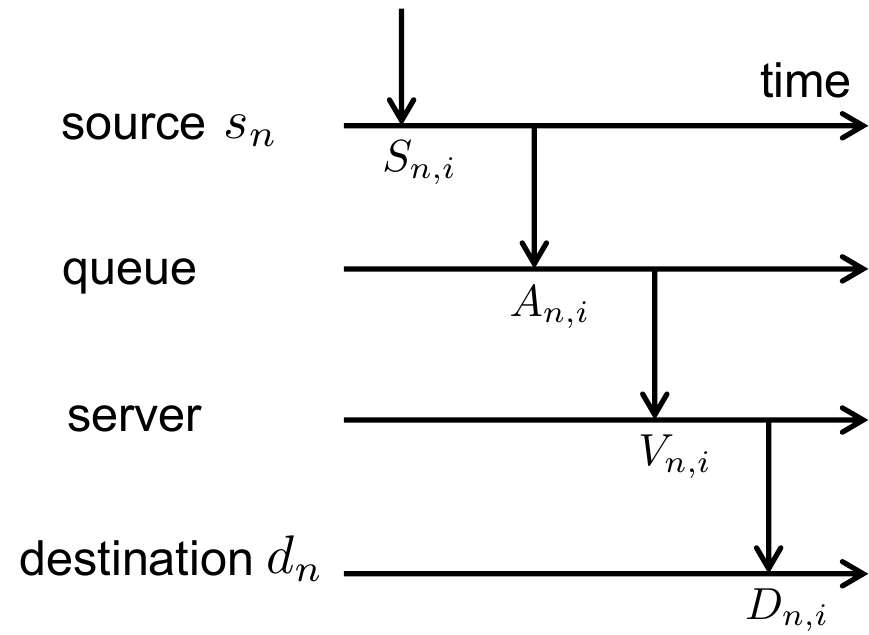} \caption{An illustration of $S_{n,i}$, $A_{n,i}$, $V_{n,i}$, and $D_{n,i}$.}
% work--efficiency ordering holds for any priorities of the jobs.
\label{fig_times1} 
\vspace{-5mm}
\end{figure}

Notice that the age $\age_{n}(t)$ in \eqref{eq_age} is determined by the packets that have been delivered to the destination $d_n$ by time $t$. To establish near age optimality, we consider an alternative age metric call the \emph{Age of Served Information}, which is determined by the packets that have started service by time $t$: Let $V_{n,i}$ denote the time that the $i$-th packet of flow $n$ is assigned to a server, i.e., the service starting time of the $i$-th packet of flow $n$, which is shown in Fig. \ref{fig_times1}. 
By definition, one can get
$S_{n,i}\leq A_{n,i}\leq V_{n,i}\leq D_{n,i}$.
The \emph{Age of Served Information} of flow $n$ is defined as
\begin{align}\label{eq_age_served}
\Xi_{n} (t) \!=\! t\! -\! \max\{S_{n,i}\!:\! V_{n,i} \leq t, i\!=\!1,2,\ldots\},
\end{align}
which is the time difference between the current time $t$ and the generation time of the freshest packet that has started service by time $t$.
As shown in Fig. \ref{fig_times2},   $\Xi_{n} (t)\leq \age_{n}(t)$.
Let $\bm{\Xi}(t)=(\Xi_{1} (t),\ldots,\Xi_{N} (t))$ denote the Age of Served Information vector at time $t$.

\begin{figure}
\centering 
\includegraphics[width=0.25\textwidth]{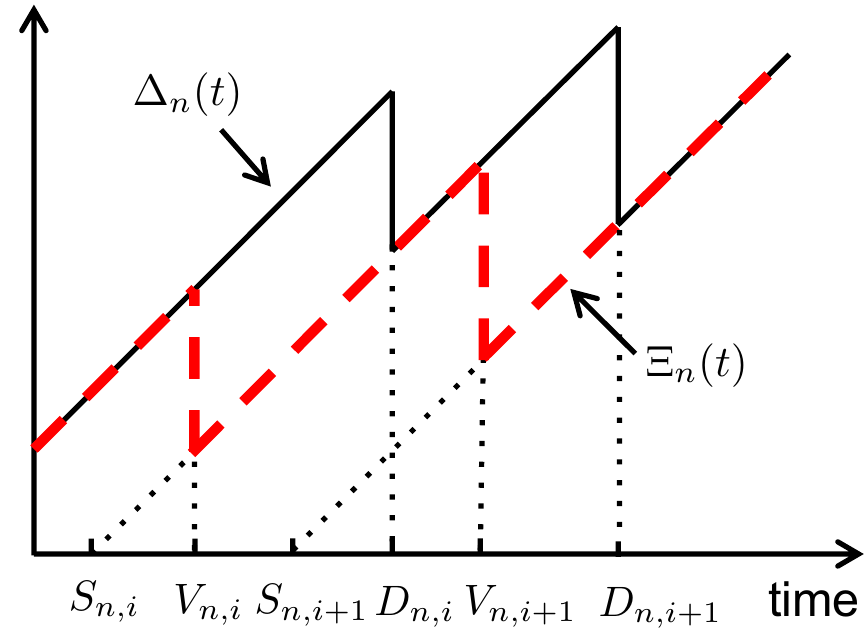} \caption{The Age of Served Information $\Xi_{n} (t)$ as a lower bound of  age $\age_{n}(t)$.}
\vspace{-5mm}
% work--efficiency ordering holds for any priorities of the jobs.
\label{fig_times2} 
\end{figure} 
We propose a new scheduling discipline called \emph{Maximum Age of Served Information First (MASIF)}, in which 
\emph{the flow with the maximum Age of Served Information is served the first, with ties broken arbitrarily}. Using this discipline, we define the following scheduling policy:

\begin{definition} \emph{Maximum Age of Served Information first, Last Generated First Served (MASIF-LGFS) policy:} In this policy, the last generated packet from the flow with the maximum \emph{Age of Served Information} is served the first among all packets of all flows, with ties broken arbitrarily. \end{definition}

{\blue In some previous studies, e.g., \cite{IgorAllerton2016,HsuTWC2017,CostaCodreanuEphremides_TIT}, it was proposed to discard old packets and only store and send the freshest one. While this technique can reduce the age, in many applications such as social updates, news seeds, and stock trading, some old packets with earlier generation times are still quite useful and are needed to be sent to the destinations. Next, we will show that the non-preemptive MASIF-LGFS policy, which does not discard old packets, is near age-optimal. Hence, the additional age reduction provided by discarding old packets in the non-preemptive MASIF-LGFS policy is not large. In order to establish this result, we first show that the age of served information of the non-preemptive MASIF-LGFS policy provides a uniform age lower bound for all non-preemptive and causal policies.}

\begin{theorem}\label{thm3}
If (i) the packet generation and arrival times are \emph{synchronized} across the $N$ flows and (ii) the packet service times are NBU and \emph{i.i.d.} across the servers and time, then it holds that for all $\mathcal{I}$, all $p_t \in\mathcal{P}_{\text{sym}}$, and all $\pi\in\Pi_{np}$ 
\begin{align}\label{thm3eq1}
&[\{p_t \circ\bm{\Xi}_{\text{non-prmp, MASIF-LGFS}}(t), t\geq 0\}\vert\mathcal{I}] \nonumber\\
\leq_{\text{st}} & [\{p_t \circ\bm{\Delta}_\pi(t), t\geq 0\}\vert\mathcal{I}],
\end{align}
or equivalently, for all $\mathcal{I}$, all $p_t \in\mathcal{P}_{\text{sym}}$, and all non-decreasing functional $\phi:\mathbb{V}\mapsto\mathbb{R}$
 \begin{align}\label{thm3eq2}
&\mathbb{E}\left[\phi (\{p_t \circ\bm{\Xi}_{\text{prmp, MAF-LGFS}}(t),t\geq 0\})\vert\mathcal{I}\right] \nonumber\\
\leq & \min_{\pi\in\Pi_{np}} \mathbb{E}\left[\phi (\{p_t \circ\bm{\Delta}_\pi(t),t\geq0\})\vert\mathcal{I}\right] \nonumber\\
\leq & \mathbb{E}\left[\phi (\{p_t \circ\bm{\age}_{\text{prmp, MAF-LGFS}}(t),t\geq 0\})\vert\mathcal{I}\right],
\end{align}
provided that the expectations in \eqref{thm3eq2} exist.
%or equivalently, for all $\mathcal{I}$, all $p \in\mathcal{P}_{\text{sym}}$, and all non-decreasing functional $\phi:\mathbb{V}\mapsto\mathbb{R}$
% \begin{align}\label{thm1eq2}
%&\mathbb{E}\left[\phi (\{p \circ\bm{\Delta}_{\text{prmp, MAF-LGFS}}(t),t\geq 0\})\vert\mathcal{I}\right] \nonumber\\
%= & \min_{\pi\in\Pi} \mathbb{E}\left[\phi (\{p \circ\bm{\Delta}_\pi(t),t\geq0\})\vert\mathcal{I}\right],
%\end{align}
%provided that the expectations in \eqref{thm1eq2} exist.
\end{theorem}
\begin{proof}[Proof idea]
Under synchronized packet generations and arrivals, the non-preemptive MASIF-LGFS policy satisfies: When a packet starts service, the flow with the maximum Age of Served Information before the service starts will have the minimum Age of Served Information  among the $N$ flows once the service starts. Theorem \ref{thm3} is proven by using this idea and  the sample-path method developed in \cite{sun2016delay,sun2017delay}. {\blue We note that the sample-path method in \cite{sun2016delay,sun2017delay} is the key for addressing the challenge in multi-flow, multi-server scheduling.} 
\ifreport
See Appendix \ref{app2} and \cite{sun2016delay,sun2017delay} for the details.
\else
See our technical report \cite{SunMultiFlow18} and \cite{sun2016delay,sun2017delay} for the details.
\fi
\end{proof}

Hence, the non-preemptive MASIF-LGFS policy is near age-optimal in the sense of \eqref{thm3eq2}. In particular, for the average age of the $N$ flows  in \eqref{eq_avgage} (i.e., $p_t = p_{\text{avg}} $), we can obtain
\begin{theorem}\label{thm4}
Under the conditions of Theorem \ref{thm3}, it holds that for all $\mathcal{I}$
\begin{align}\label{eq_gap}
\min_{\pi\in\Pi_{np}}\! [\bar{\age}_{ \pi}|\mathcal{I}] \!\leq\! [\bar{\age}_{\text{non-prmp, MASIF-LGFS}}|\mathcal{I}]\!\leq\! \min_{\pi\in\Pi_{np}}\! [\bar{\age}_{ \pi}|\mathcal{I}] \!+\! \mathbb{E}[X],\nonumber
\end{align}
where $[\bar{\age}_{ \pi}|\mathcal{I}] = \lim\sup_{T\rightarrow \infty} \frac{1}{T} \mathbb{E}[\int_0^T p_{\text{avg}} \circ\bm{\Delta}_\pi (t) dt| \mathcal{I}]$ is the expected time-average of the average age of the $N$ flows, and $\mathbb{E}[X]$ is the mean service time of one packet.
\end{theorem}

The proof of Theorem \ref{thm4} is  similar to that of Theorem 4 in \cite{BedewyJournal2017} and hence is omitted here. By Theorem \ref{thm4}, the  average age of the non-preemptive MASIF-LGFS policy is within an additive gap from the optimum, and the gap $\mathbb{E}[X]$ is invariant of the packet arrival and generation times $\mathcal{I}$, the number of flows $N$, and the number of servers $M$.

\begin{figure}
\centering 
\includegraphics[width=0.35\textwidth]{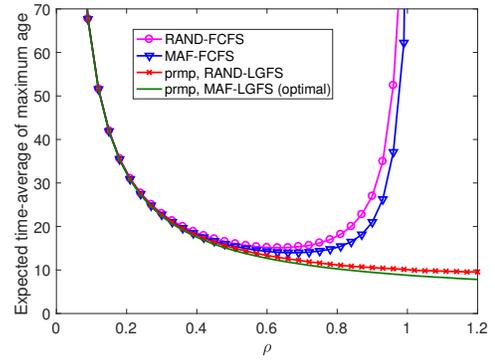} 

\caption{Expected time-average of the maximum age of 3 flows in a system with a single server and \emph{i.i.d.} exponential service times.}
% work--efficiency ordering holds for any priorities of the jobs.
\label{fig_simulation1} 
\end{figure} 

%Hence, under the conditions of Theorem \ref{thm1}, for all  $\mathcal{I}$ and all symmetric age penalty functions in $\mathcal{P}_{\text{sym}}$, the preemptive MAP policy is \emph{age-optimal} in terms of \eqref{thm1eq1} and \eqref{thm1eq2} among all policies in $\Pi$. 
\section{Numerical Results}
In this section, we evaluate the age performance of several multi-flow online scheduling policies. These scheduling policies are defined by combining the flow selection disciplines $\{$MAF, MASIF, RAND$\}$ and the packet selection disciplines $\{$FCFS, LGFS$\}$, where  RAND represents randomly choosing a flow among the flows with un-served packets. The packet generation times $S_i$ follow a Poisson process with rate $\lambda$, and the time difference $(A_i-S_i)$ between packet generation and arrival is equal to either 0 or $4/\lambda$ with equal probability. The mean service time of each server is set as $\mathbb{E}[X]=1/\mu=1$. Hence, the traffic intensity is $\rho =  \lambda N/M$, where $N$ is the number of flows and $M$ is the number of servers.

Figure \ref{fig_simulation1} illustrates the expected time-average of the maximum age $p_{\max} (\bm\age(t))$ of 3 flows in a system with a single server and \emph{i.i.d.} exponential service times. One can see that the preemptive MAF-LGFS policy has the best age performance and its age is quite small even for $\rho>1$, in which case  the queue is actually unstable. On the other hand,  both the RAND and FCFS disciplines have much higher age. Note that there is no need for preemptions under the FCFS discipline.  Figure \ref{fig_simulation2} plots the expected time-average of the average age $p_{\text{avg}} (\bm\age(t))$ of 50 flows in a system with 3 servers and \emph{i.i.d.} NBU service times. In particular, the  service time $X$ follows the following shifted exponential distribution:
\begin{align}
\Pr[X>x] = \left\{\begin{array}{l l}1,&\text{if}~x<\frac{1}{3};\\
\exp[-\frac{3}{2}(x-\frac{1}{3})],&\text{if}~x\geq \frac{1}{3}.
\end{array}\right.
\end{align}
One can observe that the non-preemptive MASIF-LGFS policy is better than the other policies, and is quite close to the age lower bound where the gap from the lower bound is no more than the mean service time $\mathbb{E}[X]=1$. {\blue One interesting observation is that the non-preemptive MASIF-LGFS policy is better than the non-preemptive MAF-LGFS policy for NBU service times. The reason behind this  is as follows: When multiple servers are idle in the non-preemptive MAF-LGFS policy, these servers are assigned to process multiple packets from the flow with the maximum age (say flow $n$). This reduces the age of flow $n$, but at a cost of postponing the service of the flows with the second and third maximum ages. On the other hand, in the non-preemptive MASIF-LGFS policy, once a packet from the flow with the maximum age of served information  (say flow $m$) is assigned to a server, the age of served information of flow $m$ drops greatly and the next server will be assigned to process the flow with second maximum age of served information. 
As shown in \cite{sun2016delay,sun2017delay}, using multiple parallel servers to process different flows is beneficial for NBU service times. The behavior of non-preemptive MASIF-LGFS policy is similar to the maximum matching scheduling algorithms, e.g., \cite{Joo:2009,Ji2014} for time-slotted systems, where multiple servers are assigned to different flows in each time-slot. One difference is that the non-preemptive MASIF-LGFS policy can even operate in continuous-time systems, but the maximum matching algorithms cannot. 
%This phenomenon suggests that the MASIF discipline deserves further investigation, which will be a research task of our future studies.
 }
%These numerical results are in accordance with our theoretical analysis in Section \ref{sec_analysis}. 

\begin{figure}
\centering 
\includegraphics[width=0.35\textwidth]{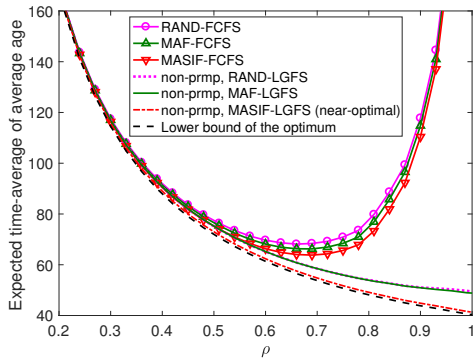} 
\caption{Expected time-average of the average age of 50 flows in a system with 3 servers and \emph{i.i.d.} NBU service times.}
% work--efficiency ordering holds for any priorities of the jobs.
\label{fig_simulation2} 
\end{figure} 

\section{Conclusion and Future Work}\label{sec_conclusion}
We have developed online scheduling policies and shown they are (near) optimal for minimizing the age of information in multi-flow, multi-server systems. Similar with \cite{BedewyJournal2017}, the results in this paper can be generalized to consider packet replications over multiple servers. {\blue In addition, similar to the results in \cite{sun2016delay,sun2017delay}, Theorem \ref{thm3} and Theorem \ref{thm4} can be generalized to the case that  the servers have different NBU service time distributions.}
Other future research directions include asynchronized packet arrivals, packet transmissions with errors, and multi-flow updates in multi-hop networks.

%A policy $P\in\Pi$ is said to be \emph{age-optimal in stochastic ordering} for minimizing the age metric process $\{p \circ\bm{\Delta}_\pi(t), t\geq 0\}$ within the policy space $\Pi$, if for all $\pi\in\Pi$
%\begin{align}\label{eq_optimal}
%\{p \circ\bm{\Delta}_P(t), t\geq 0\} \leq_{\text{st}} \{p \circ\bm{\Delta}_\pi(t), t\geq 0\},
%\end{align}
%or equivalently, for all non-decreasing functional $\phi:\mathbf{V}\mapsto\mathbb{R}$
%\begin{align}\label{eq_optimal1}
%\mathbb{E}\left[\phi (p \circ\bm{\Delta}_P)\right] =  \min_{\pi\in\Pi} \mathbb{E}\left[\phi (p \circ\bm{\Delta}_\pi)\right]
%\end{align}
%provided the expectations in \eqref{eq_optimal1} exist, 
%\input{sec_solution}
%\input{sec_examples}
\bibliographystyle{IEEEtran}
\bibliography{ref,ref1,sueh}
% !TEX root = ./heterogeneous_servers.tex
\appendices
%\appendix

\section{Proof of Theorem \ref{thm1}}\label{app1}
We first establish two lemmas that are useful to prove Theorem \ref{thm1}. 
Let the age vector $\bm\Delta_{\pi}(t)$ denote the \emph{system state} of policy $\pi$ at time $t$ and $\{\bm\Delta_{\pi}(t),t\geq 0\}$ denote the \emph{state process} of policy $\pi$. For notational simplicity, let policy $P$ represent the preemptive MAF-LGFS policy. Using
the memoryless property of exponential distribution, we can
obtain the following coupling lemma:
\begin{lemma}\emph{(Coupling Lemma)}\label{coupling}
For any given $\mathcal{I}$, consider policy $P$ and any \emph{work-conserving} policy $\pi\in \Pi$. If (i) there is a single server ($M=1$) and (ii) the packet service times are exponentially distributed and \emph{i.i.d.} across  time,   
 then there exist policy $P_1$ and  policy $\pi_1$ in the same probability space which satisfy the same scheduling disciplines with policy $P$ and policy $\pi$, respectively,  such that 
\begin{itemize}
\itemsep0em 
\item[1.] The state process $\{\bm\Delta_{P_1}(t),t\geq 0\}$ of policy $P_1$ has the same distribution with the state process $\{\bm\Delta_{P}(t),t\geq 0\}$ of policy $P$,
\item[2.] The state process $\{\bm\Delta_{\pi_1}(t),t\geq 0\}$ of policy $\pi_1$ has the same distribution with the state process $\{\bm\Delta_{\pi}(t),t\geq 0\}$  of policy $\pi$,
\item[3.] If a packet is delivered at time $t$ in policy $P_1$ as $\bm\Delta_{P_1}(t)$ evolves, then almost surely, a packet is delivered at time $t$ in policy $\pi_1$ as $\bm\Delta_{\pi_1}(t)$ evolves; and vice versa. 
%whenever there exist unassigned tasks in the queue,
\end{itemize} 
\end{lemma}
\ifreport
\begin{proof}
Note that all policies have identical arrival processes, and the service times are  memoryless. Following the inductive construction used in the proof of Theorem 6.B.3 in \cite{StochasticOrderBook}, one can construct the packet deliveries one by one in policy $P_1$ and policy $\pi_1$ to prove this lemma. The details are omitted. 
\end{proof}
\else
\begin{proof}
See our technical report \cite{SunMultiFlow18}.
\end{proof}
\fi

We will compare policy $P_1$ and policy $\pi_1$ on a sample path by using the following lemma: 

\begin{lemma} \emph{(Inductive Comparison)}\label{lem2}
Under the conditions of Lemma \ref{coupling}, suppose that a packet is delivered in policy $P_1$ and a packet is delivered in policy $\pi_1$  at the same time $t$. The system state  of policy $P_1$ is $\bm\Delta_{P_1}$ before the packet delivery, which becomes $\bm\Delta_{P_1}'$ after the packet delivery. The system state  of policy $\pi_1$ is $\bm\Delta_{\pi_1}$ before the packet delivery, which becomes $\bm\Delta_{\pi_1}'$ after the packet delivery. If the packet generation and arrival times are \emph{synchronized} across the $N$ flows and
\begin{equation}\label{hyp1}
\Delta_{[i],P_1} \leq \Delta_{[i],\pi_1},~i=1,\ldots,N,
\end{equation}
then
\begin{equation}\label{law6}
\Delta_{[i],P_1}' \leq \Delta_{[i],\pi_1}',~i=1,\ldots,N.
\end{equation}  
\end{lemma}

\ifreport
\begin{proof}
For synchronized packet generation and arrivals, let $W(t) = \max\{S_i: A_i \leq t\}$ 
be the time-stamp of the freshest packet of each flow that has arrived to the queue by time $t$. At time $t$, because no packets that has arrived is generated later than $W(t)$, we can obtain
\begin{align}%\label{eq_proof_1}
\Delta_{i,P_1} \geq\Delta_{i,P_1}' \geq t-W(t),~i=1,\ldots,N,\nonumber\\
\Delta_{i,\pi_1} \geq\Delta_{i,\pi_1}' \geq t-W(t),~i=1,\ldots,N.\label{eq_proof_2}
\end{align} 

Because there is only one server and policy $P_1$ follows the same scheduling discipline with the preemptive MAF-LGFS policy, each delivered packet in policy $P_1$ must be from the flow with the maximum age $\Delta_{[1],P_1}$ (denoted as flow $n^*$), and the delivery packet must be the last generated packet that is time-stamped with $W(t)$. In other words, the age of flow $n^*$ is reduced from the maximum age $\Delta_{[1],P_1}$ to the minimum age $\Delta_{[N],P_1}'=t-W(t)$, and the ages of the other $(N-1)$ flows remain unchanged. Hence, 
\begin{align}\label{eq_proof_3}
\Delta_{[i],P_1}' &= \Delta_{[i+1],P_1},~i=1,\ldots,N-1,\\
\Delta_{[N],P_1}' &= t - W(t). \label{eq_proof_4}
\end{align}

In policy $\pi_1$, the delivered packet can be any packet from any flow. For all possible cases of policy $\pi_1$, it must hold that 
\begin{align}\label{eq_proof_1}
\Delta_{[i],\pi_1}' \geq \Delta_{[i+1],\pi_1},~i=1,\ldots,N-1. 
\end{align}
By combining \eqref{hyp1}, \eqref{eq_proof_3}, and \eqref{eq_proof_1}, we have
\begin{align}
\Delta_{[i],\pi_1}' \geq \Delta_{[i+1],\pi_1} \geq \Delta_{[i+1],P_1} = \Delta_{[i],P_1}',~i=1,\ldots,N-1.\nonumber
\end{align}
In addition, combining \eqref{eq_proof_2} and \eqref{eq_proof_4}, yields
\begin{align}
\Delta_{[N],\pi_1}' \geq  t-W(t) = \Delta_{[N],P_1}'.\nonumber
\end{align}
By this, \eqref{law6} is proven.
\end{proof}
\else
\begin{proof}
See our technical report \cite{SunMultiFlow18}.
\end{proof}
\fi

%\begin{lemma}
%Consider two $N$-dimensional vectors $\bm{x}$ and $\bm{y}$. If $\bm{x}\leq\bm{y}$, then $x_{[i]} \leq y_{[i]}$ for all $i=1,\ldots,N$.
%\end{lemma}
%\begin{proof}
%For each $i=1,\ldots,N$, there exist $i$ elements $x_{[1]}, \ldots, x_{[i]}$ in $\bm{x}$ which are no smaller than $x_{[i]}
%$. This, together with $\bm{x}\leq\bm{y}$, tells us that at least $i$ elements in $\bm{y}$ are no smaller than $x_{[i]}
%$. Because $y_{[i]}$ is the $i$-th largest element in $\bm{y}$, $x_{[i]} \leq y_{[i]}$. This completes the proof.
%\end{proof}
Now we are ready to prove Theorem \ref{thm1}.
\begin{proof}[Proof of Theorem \ref{thm1}]
%See Appendix \ref{app1}.
Consider any work-conserving policy $\pi\in\Pi$. By Lemma \ref{coupling}, there exist policy $P_1$ and policy $\pi_1$
satisfying the same scheduling disciplines with policy $P$ and policy $\pi$, respectively, and the packet delivery times in policy $P_1$ and policy $\pi_1$ are synchronized almost surely.

For any given sample path of policy $P_1$ and policy $\pi_1$, $\bm\Delta_{P_1}(0^-) = \bm\Delta_{\pi_1}(0^-)$ at time $t=0^-$. We consider two cases:

\emph{Case 1:} When there is no packet delivery, the age of each flow grows linearly with a slope 1. 

\emph{Case 2:} When a packet is delivered, the evolution of the age is governed by  Lemma \ref{lem2}. 

By induction over time, we obtain
\begin{align}\label{eq_thm1_proof1}
\Delta_{[i],P_1} (t) \leq \Delta_{[i],\pi_1} (t),~i=1,\ldots,N,~t\geq 0.
\end{align}

For any symmetric and non-decreasing   function $p_t$, it holds from \eqref{eq_thm1_proof1} that for all sample paths and all $t\geq 0$
\begin{align}\label{eq_thm1_proof3}
&p_t\circ \bm \Delta_{P_1}(t) \nonumber\\
=& p_t(\Delta_{1,P_1} (t), \ldots, \Delta_{N,P_1} (t))\nonumber\\
=& p_t (\Delta_{[1],P_1} (t), \ldots, \Delta_{[N],P_1} (t))\nonumber\\
\leq & p_t (\Delta_{[1],\pi_1} (t), \ldots, \Delta_{[N],\pi_1} (t))\nonumber\\
=& p_t (\Delta_{1,\pi_1} (t), \ldots, \Delta_{N,\pi_1} (t))\nonumber\\
=& p_t\circ \bm \Delta_{\pi_1}(t).
\end{align}
By Lemma \ref{coupling}, the state process $\{\bm\Delta_{P_1}(t),t\geq 0\}$ of policy $P_1$ has the same distribution with the state process $\{\bm\Delta_{P}(t),t\geq 0\}$ of policy $P$;
the state process $\{\bm\Delta_{\pi_1}(t),t\geq 0\}$ of policy $\pi_1$ has the same distribution with the state process $\{\bm\Delta_{\pi}(t),t\geq 0\}$  of policy $\pi$. By \eqref{eq_thm1_proof3} and Theorem 6.B.30 in \cite{StochasticOrderBook}, \eqref{thm1eq1} holds for all work-conserving policy $\pi\in\Pi$. 

For non-work-conserving policies $\pi$, because the service times are exponentially distributed and \emph{i.i.d.} across servers and time, server idling only postpones the delivery times of the packets. One can construct a coupling to show that for any non-work-conserving policy $\pi$, there exists a work-conserving policy $\pi'$ whose age process is smaller than that of policy $\pi$ in stochastic ordering; the details are omitted. 
%the age of 
%Therefore, the age under non-work-conserving policies will be greater. 
As a result, \eqref{thm1eq1} holds for all policies $\pi\in\Pi$.

Finally, the equivalence between \eqref{thm1eq1} and \eqref{thm1eq2} follows from \eqref{eq_order}. This completes the proof.
\end{proof}

\ifreport

\section{Proof of Theorem \ref{thm3}}\label{app2}
This proof is motivated by the sample-path  method developed in \cite{sun2016delay,sun2017delay} for near delay-optimal scheduling  in multi-server queueing systems.

We first provide two useful lemmas. 
Let $(\bm\Delta_{\pi}(t),\bm\Xi_{\pi}(t))$ denote the \emph{system state} of policy $\pi$ at time $t$ and $\{(\bm\Delta_{\pi}(t),\bm\Xi_{\pi}(t)),t\geq 0\}$ denote the \emph{state process} of policy $\pi$. For notational simplicity, let policy $P$ represent the non-preemptive MASIF-LGFS policy.

In single-server queueing systems, the following \emph{work conservation law} (or its generalizations)
plays an important role in the analysis of scheduling performance:
At any time, the expected total amount of time for completing the packets in the queue is invariant among all work-conserving policies \cite{Leonard_Kleinrock_book,Jose2010,Gittins:11}. However, the work conservation law does not hold in multi-server queueing systems, where it is difficult to fully utilize all the servers to process the packets. Specifically, it may happen that some servers are busy while the remaining servers are idle, where the idleness leads to inefficient packet service and a performance gap from the optimum. 
In the sequel, we introduce an ordering to compare the efficiency of packet service in different policies in a near-optimal sense, which is called \emph{weak work-efficiency ordering}.\footnote{Two work-efficiency orderings were used in \cite{sun2016delay,sun2017delay} to study (near) delay-optimal online scheduling in multi-server queueing systems.}

\begin{definition} \label{def_order} \emph{Weak Work-efficiency Ordering \cite{sun2016delay,sun2017delay}:}
For any given $\mathcal{I}$ and a sample path of two policies $\pi_1,\pi_2\in\Pi_{np}$, policy $\pi_1$ is said to be \emph{weakly more work-efficient than} policy $\pi_2$, if the following assertion is true:
%there is a one-to-one correspondence between the packets executed in policy $P$ and policy $\pi$ such that, if
\emph{For each packet $j$ executed in policy $\pi_2$, if
\begin{itemize}
\item[1.] In policy $\pi_2$, packet $j$ starts service at time $\tau$ and completes service at time $\nu$ ($\tau\leq \nu$), 
\item[2.] In policy $\pi_1$, the queue is not empty during $[\tau,\nu]$, 
\end{itemize}
then there always exists one corresponding packet $j'$ in policy $\pi_1$ which starts service during $[\tau,\nu]$.} \end{definition}

\begin{figure}
\centering 
\includegraphics[width=0.3\textwidth]{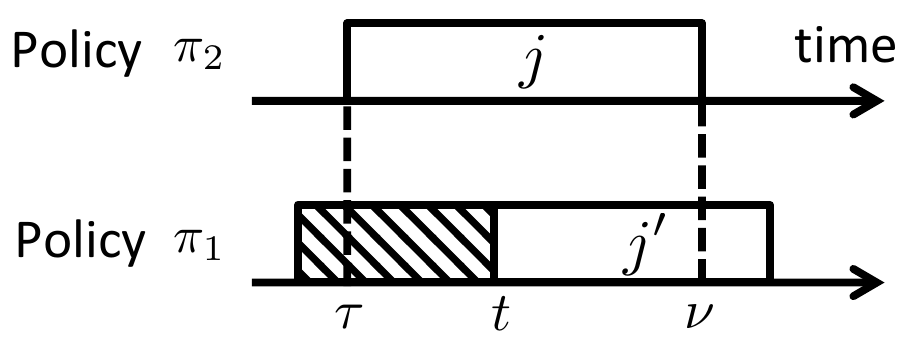} \caption{A sample-path illustration of the weak work-efficiency ordering between two policies $\pi_1$ and $\pi_2$, where the service duration of a packet is indicated by a rectangle, without specifying which server  is used to process the packet. 
%If a packet is replicated on multiple servers, then the service duration of this packet starts from the earliest time that one replica of this packet enters a server, until one replica of this packet is completed. 
In this figure, a packet $j$ starts service at time $\tau$ and completes service at time $\nu$ in policy $\pi_2$, a corresponding packet $j'$ that starts service during $[\tau,\nu]$ in policy $\pi_1$.}
% work--efficiency ordering holds for any priorities of the jobs.
\label{Work_Efficiency_Ordering} 
\end{figure} 

An sample-path illustration of the weak work-efficiency ordering is provided in Fig. \ref{Work_Efficiency_Ordering}. In particular, if policy $\pi_1$ is {weakly more work-efficient than} policy $\pi_2$, then each packet in policy $\pi_1$ must start service during the service duration $[\tau,\nu]$ of its corresponding packet in policy $\pi_2$, or the queue is empty during  $[\tau,\nu]$ in policy $\pi_1$. Note that the weak work-efficient ordering does not require to specify which server is used to  serve each packet. 

%This is the key feature that enables us to establish a tight age lower bound. 

The following coupling lemma was established in \cite{sun2017delay} by using the property of NBU distributions and the fact that policy $P$ (i.e., the non-preemptive MASIF-LGFS policy) is work-conserving:
\begin{lemma}\emph{(Coupling Lemma)} \cite[Lemma 2]{sun2017delay}\label{thm3lem_coupling}
Consider two policies $P,\pi\in \Pi_{np}$. If (i) policy $P$ is work-conserving and (ii) the packet service times are NBU, independent across the servers, and \emph{i.i.d.} across  the packets assigned to the same server, then there exist policy $P_1$ and  policy $\pi_1$ in the same probability space which satisfy the same scheduling disciplines with policy $P$ and policy $\pi$, respectively,  such that 
\begin{itemize}
\itemsep0em 
\item[1.] The state process $\{(\bm\Delta_{P_1}(t),\bm\Xi_{P_1}(t)),t\geq 0\}$ of policy $P_1$ has the same distribution with the state process $\{(\bm\Delta_{P}(t),\bm\Xi_{P}(t)),t\geq 0\}$ of policy $P$,
\item[2.] The state process $\{(\bm\Delta_{\pi_1}(t),\bm\Xi_{\pi_1}(t)),t\geq 0\}$ of policy $\pi_1$ has the same distribution with the state process $\{(\bm\Delta_{\pi}(t),\bm\Xi_{\pi}(t)),t\geq 0\}$  of policy $\pi$,
\item[3.] Policy $P_1$ is weakly more work-efficient than policy $\pi_1$ with probability one. 
%whenever there exist unassigned packets in the queue,
\end{itemize} 
\end{lemma}
The proof of Lemma \ref{thm3lem_coupling} is provided in \cite{sun2017delay}.
Note that Lemma \ref{thm3lem_coupling} holds even if  policy $P$ is replaced by any non-preemptive work-conserving policy.

We will compare the age lower bound of policy $P_1$ and the age of policy $\pi_1$ on a sample path by using the following lemma:

\begin{lemma} \emph{(Inductive Comparison)}\label{thm3lem2}
Under the conditions of Lemma \ref{coupling}, suppose that a packet starts service in policy $P_1$ and a packet completes service (i.e., delivered to the destination) in policy $\pi_1$  at the same time $t$. The system state  of policy $P_1$ is $(\bm\Delta_{P_1},\bm\Xi_{P_1})$ before the service starts, which becomes $(\bm\Delta_{P_1}',\bm\Xi_{P_1}')$ after the service starts. The system state  of policy $\pi_1$ is $(\bm\Delta_{\pi_1},\bm\Xi_{\pi_1})$ before the service completes, which becomes $(\bm\Delta_{\pi_1}',\bm\Xi_{\pi_1}')$ after the service completes.
 If the packet generation and arrival times are \emph{synchronized} across the $N$ flows and
\begin{equation}\label{thm3hyp1}
 \Xi_{[i],P_1} \leq \Delta_{[i],\pi_1},~i=1,\ldots,N,
\end{equation}
then
\begin{equation}\label{thm3law6}
\Xi_{[i],P_1}' \leq \Delta_{[i],\pi_1}',~i=1,\ldots,N.
\end{equation}  
\end{lemma}

\begin{proof}
For synchronized packet generation and arrivals,  let $W(t) = \max\{S_i: A_i \leq t\}$ 
be the time-stamp of the freshest packet of each flow that has arrived to the queue by time $t$. At time $t$, because no packets that has arrived is generated later than $W(t)$, we can obtain
\begin{align}%\label{eq_proof_1}
\Xi_{i,P_1} \geq\Xi_{i,P_1}' \geq t-W(t),~i=1,\ldots,N,\nonumber\\
\Delta_{i,\pi_1} \geq\Delta_{i,\pi_1}' \geq t-W(t),~i=1,\ldots,N.\label{thm3eq_proof_2}
\end{align} 

Because there is only one server and policy $P_1$ follows the same scheduling discipline with the non-preemptive MASIF-LGFS policy, each  packet starts service in policy $P_1$ must be from the flow with the maximum  age of served information $\Xi_{[1],P_1}$ (denoted as flow $n^*$), and the delivery packet must be the last generated packet that is time-stamped with $W(t)$. In other words, the age of served information of flow $n^*$ is reduced from the maximum age of served information $\Xi_{[1],P_1}$ to the minimum age of served information $\Xi_{[N],P_1}'=t-W(t)$, and the ages of served information of the other $(N-1)$ flows remain unchanged. Hence, 
\begin{align}\label{thm3eq_proof_3}
\Xi_{[i],P_1}' &= \Xi_{[i+1],P_1},~i=1,\ldots,N-1,\\
\Xi_{[N],P_1}' &= t - W(t). \label{thm3eq_proof_4}
\end{align}

In policy $\pi_1$, the delivered packet can be any packet from any flow. For all possible cases of policy $\pi_1$, it must hold that 
\begin{align}\label{thm3eq_proof_1}
\Delta_{[i],\pi_1}' \geq \Delta_{[i+1],\pi_1},~i=1,\ldots,N-1. 
\end{align}
By combining \eqref{thm3hyp1}, \eqref{thm3eq_proof_3}, and \eqref{thm3eq_proof_1}, we have
\begin{align}
\Delta_{[i],\pi_1}' \geq \Delta_{[i+1],\pi_1} \geq \Xi_{[i+1],P_1} = \Xi_{[i],P_1}',~i=1,\ldots,N-1.\nonumber
\end{align}
In addition, combining \eqref{thm3eq_proof_2} and \eqref{thm3eq_proof_4}, yields
\begin{align}
\Delta_{[N],\pi_1}' \geq  t-W(t) = \Xi_{[N],P_1}'.\nonumber
\end{align}
By this, \eqref{thm3law6} is proven.
\end{proof}
Now we are ready to prove Theorem \ref{thm3}.
\begin{proof}[Proof of Theorem \ref{thm3}]
Consider any  policy $\pi\in\Pi_{np}$. By Lemma \ref{thm3lem_coupling}, there exist policy $P_1$ and policy $\pi_1$
satisfying the same scheduling disciplines with policy $P$ and policy $\pi$, respectively, and policy $P_1$ is weakly more work-efficient than policy $\pi_1$ with probability one.

Next, we construct a policy $\pi_1'$ in the same probability space with policy  $P_1$ and policy $\pi_1$: Let $\bm\Delta_{\pi_1'}(0^-) = \bm\Xi_{P_1}(0^-) = \bm\Delta_{\pi_1}(0^-) $ at time $t=0^-$. For each pair of corresponding packet $j$ and packet $j'$ mentioned in the definition of the weak work-efficiency ordering, if 
%After this modification, the following three properties are true:
\begin{itemize}
\item In policy $\pi_1$, packet $j$ starts service at time $\tau$ and completes service at time $\nu$ ($\tau\leq \nu$),
\item In policy $P_1$, the queue is not empty  during $[\tau,\nu]$,
\item In policy $P_1$, the corresponding packet $j'$ starts service at time $t\in[\tau,\nu]$,
\end{itemize}
then in policy $\pi_1'$, packet $j$ starts service at time $\tau$ and completes service at time $t\in[\tau,\nu]$, as illustrated in Fig. \ref{Work_Efficiency_Ordering2}. Policy $\pi_1'$ satisfies the following two useful properties: 
\begin{figure}
\centering 
\includegraphics[width=0.3\textwidth]{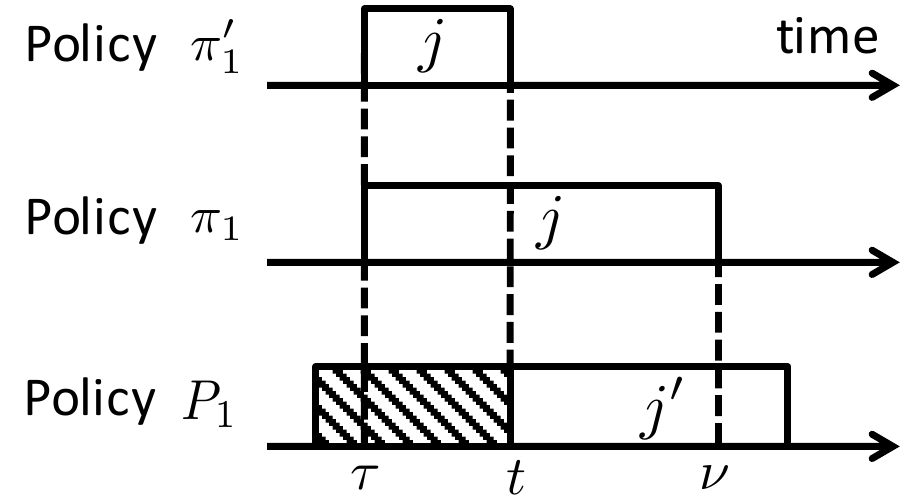} \caption{A illustration of the construction of policy $\pi'_1$. The service completion time of packet $j$ in policy $\pi_1'$ is smaller than the service completion time of packet $j$ in policy $\pi$, and is equal to the service starting time of packet $j'$. }
% work--efficiency ordering holds for any priorities of the jobs.
\label{Work_Efficiency_Ordering2} 
\end{figure}

First, when the queue is not empty in policy $P_1$, the delivery time of each packet in policy $\pi_1'$ is earlier than that in policy $\pi$. In particular, the delivery time of packet $j$ is $t$ in policy $\pi_1'$, which is earlier than $\nu$, i.e., the delivery time of packet $j$ in policy $\pi_1$. Hence, 
\begin{align}\label{eq_thm3_proof2}
\bm\Delta_{\pi_1'}(t) \leq \bm\Delta_{\pi_1}(t), ~t\in[0,\infty)
\end{align}
holds with probably one.

Second,  when the queue is not empty in policy $P_1$, the packet delivery times in policy $\pi_1'$ is synchronized with the service starting times in policy $P_1$. 
We now use this property to  show that with probability one 
\begin{align}\label{eq_thm3_proof1}
\Xi_{[i],P_1} (t) \leq \Delta_{[i],\pi_1'} (t),~i=1,\ldots,N,~t\geq 0.
\end{align}
For any given sample path of policy $P_1$ and policy $\pi_1'$, $\bm\Xi_{P_1}(0^-) = \bm\Delta_{\pi_1'}(0^-)$ at time $t=0^-$. Let us consider three cases:

\emph{Case 1:} When there is no packet delivery in policy $\pi_1'$, the age and the age of served information of each flow grows linearly with a slope 1. 

\emph{Case 2:} When there is a packet delivery in policy $\pi_1'$ and the queue is not empty in policy $P_1$, the evolution of the system state is governed by  Lemma \ref{thm3lem2}. 

\emph{Case 3:} When there is a packet delivery in policy $\pi_1'$ and the queue is empty (all packets are delivered or under service) in policy $P_1$, \eqref{eq_thm3_proof1} holds naturally as all packets have started services in policy $P_1$. 

By induction over time and considering these three cases, \eqref{eq_thm3_proof1} is proven.

Next, for any symmetric and non-decreasing   function $p_t$, it holds from \eqref{eq_thm3_proof2} and \eqref{eq_thm3_proof1} that for all sample paths and all $t\geq 0$
\begin{align}\label{eq_thm3_proof3}
&p_t\circ \bm \Xi_{P_1}(t) \nonumber\\
=& p_t(\Xi_{1,P_1} (t), \ldots, \Xi_{N,P_1} (t))\nonumber\\
=& p_t (\Xi_{[1],P_1} (t), \ldots, \Xi_{[N],P_1} (t))\nonumber\\
\leq & p_t (\Delta_{[1],\pi_1'} (t), \ldots, \Delta_{[N],\pi_1'} (t))\nonumber\\
=& p_t (\Delta_{1,\pi_1'} (t), \ldots, \Delta_{N,\pi_1'} (t))\nonumber\\
=& p_t\circ \bm \Delta_{\pi_1'}(t)\nonumber\\
\leq & p_t\circ \bm \Delta_{\pi_1}(t).
\end{align}
By Lemma \ref{thm3lem_coupling}, the state process $\{\bm\Delta_{P_1}(t),t\geq 0\}$ of policy $P_1$ has the same distribution with the state process $\{\bm\Delta_{P}(t),t\geq 0\}$ of policy $P$;
the state process $\{\bm\Delta_{\pi_1}(t),t\geq 0\}$ of policy $\pi_1$ has the same distribution with the state process $\{\bm\Delta_{\pi}(t),t\geq 0\}$  of policy $\pi$. By \eqref{eq_thm3_proof3} and  Theorem 6.B.30 in \cite{StochasticOrderBook}, \eqref{thm3eq1} holds for all  policy $\pi\in\Pi_{np}$. Finally, the equivalence between \eqref{thm3eq1} and \eqref{thm3eq2} follows from \eqref{eq_order}. This completes the proof.
\end{proof}
\fi

\end{document}